\documentclass[english]{llncs}
\usepackage[T1]{fontenc}
\usepackage[latin9]{inputenc}
\usepackage{amssymb}

\makeatletter
\usepackage{fancyhdr}
\makeatother

\usepackage{babel}
\begin{document}

\title{Decision Making for Inconsistent Expert Judgments Using Negative
Probabilities}

\titlerunning{Inconsistent Decision Making}

\author{J. Acacio de Barros}

\institute{Liberal Studies Program, San Francisco State University, 1600 Holloway
Ave., San Francisco, CA 94132}
\maketitle
\begin{abstract}
In this paper we provide a simple random-variable example of inconsistent
information, and analyze it using three different approaches: Bayesian,
quantum-like, and negative probabilities. We then show that, at least
for this particular example, both the Bayesian and the quantum-like
approaches have less normative power than the negative probabilities
one. 
\end{abstract}

\section{Introduction}

In recent years the quantum-mechanical formalism (mainly from non-relativistic
quantum mechanics) has been used to model economic and decision-making
processes (see \cite{khrennikov_ubiquitous_2010,haven_quantum_2013}
and references therein). The success of such models may originate
from several related issues. First, the quantum formalism leads to
a propositional structure that does not conform to classical logic
\cite{birkhoff_logic_1936}. Second, the probabilities of quantum
observables do not satisfy Kolmogorov's axioms \cite{weizsacker_probability_1973}.
Third, quantum mechanics describes experimental outcomes that are
highly contextual \cite{bell_einstein-podolsky-rosen_1964,bell_problem_1966,kochen_problem_1975,greenberger_going_1989,de_barros_inequalities_2000}.
Such issues are connected because the logic of quantum mechanics,
represented by a quantum lattice structure \cite{birkhoff_logic_1936},
leads to upper probability distributions and thus to non-Kolmogorovian
measures \cite{holik_discussion_2012,de_barros_probabilistic_2001,de_barros_probabilistic_2010},
while contextuality leads the nonexistence of a joint probability
distribution \cite{suppes_when_1981,suppes_collection_1996}. 

Both from a foundational and from a practical point of view, it is
important to ask which aspects of quantum mechanics are actually needed
for social-science models. For instance, the Hilbert space formalism
leads to non-standard logic and probabilities, but the converse is
not true: one cannot derive the Hilbert space formalism solely from
weaker axioms of probabilities or from quantum lattices. Furthermore,
the quantum formalism yields non-trivial results such as the impossibility
of superluminal signaling with entangled states \cite{dieks_communication_1982}.
These types of results are not necessary for a theory of social phenomena
\cite{de_barros_quantum_2009}, and we should ask what are the minimalistic
mathematical structures suggested by quantum mechanics that reproduce
the relevant features of quantum-like behavior. 

In a previous article, we used reasonable neurophysiological assumptions
to created a neural-oscillator model of behavioral Stimulus-Response
theory \cite{suppes_phase-oscillator_2012}. We then showed how to
use such model to reproduce quantum-like behavior \cite{de_barros_quantum-like_2012}.
Finally, in a subsequent article, we remarked that the same neural-oscillator
model could be used to represent a set of observables that could not
correspond to quantum mechanical observables \cite{de_barros_joint_2012},
in a sense that we later on formalize in Section \ref{sec:Quantum-Approach}.
These results suggest that one of the main quantum features relevant
to social modeling is contextuality, represented by a non-Kolmogorovian
probability measure, and that imposing a quantum formalism may be
too restrictive. This non-Kolmogorovian characteristic would come
when two contexts providing incompatible information about observable
quantities were present. 

Here we focus on the incompatibility of contexts as the source of
a violation of standard probability theory. We then ask the following
question: what formalisms are normative with respect to such incompatibility?
This question comes from the fact that, in its origin, probability
was devised as a normative theory, and not descriptive. For instance,
Richard Jeffrey \cite{jeffrey_probability_1992} explains that ``the
term 'probable' (Latin \emph{probable}) meant \emph{approvable}, and
was applied in that sense, univocally, to opinion and to action. A
probable action or opinion was one such as sensible people would undertake
or hold, in the circumstances.'' Thus, it should come as no surprise
that humans actually violate the rules of probability, as shown in
many psychology experiments. However, if a person is to be considered
``rational,'' according to Boole, he/she should follow the rules
of probability theory. 

Since inconsistent information, as above mentioned, violates the theory
of probability, how do we provide a normative theory of rational decision-making?
There are many approaches, such as Bayesian models or the Dempster-Shaffer
theory, but here we focus on two non-standard ones: quantum-like and
negative probability models. We start first by presenting a simple
case where expert judgments lead to inconsistencies. Then, we approach
this problem first with a standard Bayesian probabilistic method,
followed by a quantum model. Finally, we use negative probability
distributions as a third alternative. We then compare the different
outcomes of each approach, and show that the use of negative probabilities
seems to provide the most normative power among the three. We end
this paper with some comments.

\section{Inconsistent Information\label{sec:Inconsistent-Information}}

As mentioned, the use of the quantum formalism in the social sciences
originates from the observation that Kolmogorov's axioms are violated
in many situations \cite{khrennikov_ubiquitous_2010,haven_quantum_2013}.
Such violations in decision-making seem to indicate a departure from
a rational view, and in particular to though-processes that may involve
irrational or contradictory reasoning, as is the case in non-monotonic
reasoning. Thus, when dealing with quantum-like social phenomena,
we are frequently dealing with some type of inconsistent information,
usually arrived at as the end result of some non-classical (or incorrect,
to some) reasoning. In this section we examine the case where inconsistency
is present from the beginning. 

Though in everyday life inconsistent information abounds, standard
classical logic has difficulties dealing with it. For instance, it
is a well know fact that if we have a contradiction, i.e. $A\&\left(\neg A\right)$,
then the logic becomes trivial, in the sense that any formula in such
logic is a theorem. To deal with such difficulty, logicians have proposed
modified logical systems (e.g. paraconsistent logics \cite{da_costa_theory_1974}).
Here, we will discuss how to deal with inconsistencies not from a
logical point of view, but instead from a probabilistic one. 

Inconsistencies of expert judgments are often represented in the probability
literature by measures corresponding to the experts' subjective beliefs
\cite{genest_combining_1986}. It is frequently argued that this subjective
nature is necessary, as each expert makes statements about outcomes
that are, in principle, available to all experts, and disagreements
come not from sampling a certain probability space, but from personal
beliefs. For example, let us assume that two experts, Alice and Bob,
are examining whether to recommend the purchase of stocks in company
$X$, and each gives different recommendations. Such differences do
not emerge from an objective data (i.e. the actual future prices of
$X$), but from each expert's interpretations of current market conditions
and of company $X$. In some cases the inconsistencies are evident,
as when, say, Alice Alice recommends buy, and Bob recommends sell;
in this case the decision maker would have to reconcile the discrepancies. 

The above example provides a simple case. A more subtle one is when
the experts have inconsistent beliefs that seem to be consistent.
For example, each expert, with a limited access to information, may
form, based on different contexts, locally consistent beliefs without
directly contradicting other experts. But when we take the totality
of the information provided by all of them and try to arrive at possible
inferences, we reach contradictions. Here we want to create a simple
random-variable model that incorporates expert judgments that are
locally consistent but globally inconsistent. This model, inspired
by quantum entanglement, will be used to show the main features of
negative probabilities as applied to decision making. 

Let us start with three $\pm1$-valued random variables, $\mathbf{X}$,
$\mathbf{Y}$, and $\mathbf{Z}$, with zero expectation. If such random
variables have correlations that are too strong then there is no joint
probability distribution \cite{suppes_when_1981}. To see this, imagine
the extreme case where the correlations between the random variables
are $E\left(\mathbf{XY}\right)=E\left(\mathbf{YZ}\right)=E\left(\mathbf{XZ}\right)=-1.$
Imagine that in a given trial we draw $\mathbf{X}=1$. From $E\left(\mathbf{XY}\right)=-1$
it follows that $\mathbf{Y}=-1$, and from $E\left(\mathbf{YZ}\right)=-1$
that $\mathbf{Z}=1$. But this is in contradiction with $E\left(\mathbf{XZ}\right)=-1$,
which requires $\mathbf{Z}=-1$. Of course, the problem is not that
there is a mathematical inconsistency, but that it is not possible
to find a probabilistic sample space for which the variables $\mathbf{X}$,
$\mathbf{Y}$, and $\mathbf{Z}$ have such strong correlations. Another
way to think about this is that the the $\mathbf{X}$ measured together
with $\mathbf{Y}$ is not the same one as the $\mathbf{X}$ measured
with $\mathbf{Z}$: values of $\mathbf{X}$ depend on its context. 

The above example posits a deterministic relationship between all
random variables, but the inconsistencies persist even when weaker
correlations exist. In fact, Suppes and Zanotti \cite{suppes_when_1981}
proved that a joint probability distribution for $\mathbf{X}$, $\mathbf{Y}$,
and $\mathbf{Z}$ exists if and only if 
\begin{eqnarray}
-1 & \leq & E\left(\mathbf{XY}\right)+E\left(\mathbf{YZ}\right)+E\left(\mathbf{XZ}\right)\nonumber \\
 & \leq & 1+2\min\left\{ E\left(\mathbf{XY}\right),E\left(\mathbf{YZ}\right),E\left(\mathbf{XZ}\right)\right\} .\label{eq:joint-inequality}
\end{eqnarray}
The above case violates inequality (\ref{eq:joint-inequality}). 

Let's us now consider the example we want to analyze in detail. Imagine
$\mathbf{X}$, $\mathbf{Y}$, and $\mathbf{Z}$ as corresponding to
future outcomes in a company's stocks. For instance, $\mathbf{X}=1$
corresponds to an increase of the stock value of company $X$ in the
following day, while $\mathbf{X}=-1$ a decrease, and so on. Three
experts, Alice ($A$), Bob ($B$), and Carlos ($C$), have the following
beliefs about those stocks. Alice is an expert on companies $X$ and
$Y$, but knows little or nothing about $Z$, so she only tells us
what we don't know: her expected correlation $E_{A}\left(\mathbf{XY}\right)$.
Bob (Carlos), on the other hand, is only an expert in companies $X$
and $Z$ ($Y$ and $Z$), and he too only tells us about their correlations.
Let us take the case where 
\begin{equation}
E_{A}\left(\mathbf{XY}\right)=-1,\label{eq:xy}
\end{equation}
\begin{equation}
E_{B}\left(\mathbf{XZ}\right)=-\frac{1}{2},\label{eq:xz}
\end{equation}
\begin{equation}
E_{C}\left(\mathbf{YZ}\right)=0,\label{eq:yz}
\end{equation}
where the subscripts refer to each experts. For such case, the sum
of the correlations is $-1\frac{1}{2}$, and according to (\ref{eq:joint-inequality})
no joint probability distribution exists. Since there is no joint,
how can a rational decision-maker decide what to do when faced with
the question of how to bet in the market? In particular, how can she
get information about the joint probability, and in particular the
unknown triple moment $E\left(\mathbf{XYZ}\right)$? In the next sections
we will show how we can try to answer these questions using three
possible approaches: quantum, Bayesian, and signed probabilities.

\section{Quantum Approach\label{sec:Quantum-Approach}}

We start with a comment about the quantum-like nature of correlations
(\ref{eq:xy})-(\ref{eq:yz}). The random variables $\mathbf{X}$,
$\mathbf{Y}$, and $\mathbf{Z}$ with correlations (\ref{eq:xy})-(\ref{eq:yz})
cannot be represented by a quantum state in a Hilbert space for the
observables corresponding to $\mathbf{X}$, $\mathbf{Y}$, and $\mathbf{Z}$.
This claim can be expressed in the form of a simple proposition. 
\begin{proposition}
Let $\hat{X}$, $\hat{Y}$, and $\hat{Z}$ be three observables in
a Hilbert space ${\cal H}$ with eigenvalues $\pm1$, let them pairwise
commute, and let the $\pm1$-valued random variable $\mathbf{X}$,
\textbf{$\mathbf{Y}$}, and $\mathbf{Z}$ represent the outcomes of
possible experiments performed on a quantum system $|\psi\rangle\in{\cal H}$.
Then, there exists a joint probability distribution consistent with
all the possible outcomes of $\mathbf{X}$, \textbf{$\mathbf{Y}$},
and $\mathbf{Z}$.\end{proposition}
\begin{proof}
Because $\hat{X}$, $\hat{Y}$, and $\hat{Z}$ are observables and
they pairwise commute, it follows that their combinations, $\hat{X}\hat{Y}$,
$\hat{Y}\hat{Z}$, $\hat{X}\hat{Z}$, and $\hat{X}\hat{Y}\hat{Z}$
are also observables, and they commute with each other. For instance,
\[
\left(\hat{X}\hat{Y}\hat{Z}\right)^{\dagger}=\hat{Z}^{\dagger}\hat{Y}^{\dagger}\hat{X}^{\dagger}=\hat{X}\hat{Y}\hat{Z}.
\]
Furthermore, 
\[
[\hat{X}\hat{Y}\hat{Z},\hat{X}]=[\hat{X}\hat{Y}\hat{Z},\hat{Y}]=\cdots=[\hat{X}\hat{Y}\hat{Z},\hat{X}\hat{Z}]=0.
\]
Therefore, quantum mechanics implies that all three observables $\hat{X}$,
$\hat{Y}$, and $\hat{Z}$ can be simultaneously measured. Since this
is true, for the same state $|\psi\rangle$ we can create a full data
table with all three values of $\mathbf{X}$, \textbf{$\mathbf{Y}$},
and $\mathbf{Z}$ (i.e., no missing values), which implies the existence
of a joint. 
\end{proof}
So, how would a quantum-like model of correlations (\ref{eq:xy})--(\ref{eq:yz})
be like? The above result depends on the use of the same quantum state
$|\psi\rangle$ throughout the many runs of the experiment, and to
circumvent it we would need to use different states for the system.
In other words, if we want to use a quantum formalism to describe
the correlations (\ref{eq:xy})-(\ref{eq:yz}), a $|\psi\rangle$
would have to be selected for each run such that a different state
would be used when we measure $\hat{X}\hat{Y}$, e.g. $|\psi\rangle_{xy}$,
than when we measure $\hat{X}\hat{Z}$, e.g. $|\psi\rangle_{xz}$.
Then, the quantum description could be accomplished by the state 
\[
|\psi\rangle=c_{A}|A\rangle\otimes|\psi\rangle_{xy}+c_{B}|B\rangle\otimes|\psi\rangle_{xz}+c_{C}|C\rangle\otimes|\psi\rangle_{yz}.
\]
This state would model the correlations the following way. When Alice
makes her choice, she uses a projector into her ``state of knowledge''
$\hat{P}_{A}=|A\rangle\langle A|$, and gets the correlation $E_{A}\left(\mathbf{XY}\right)$,
and similarly for Bob and Carlos. 

In the above example, all correlations and expectations are given,
and the only unknown is the triple moment $E\left(\mathbf{XYZ}\right)$.
Furthermore, since we do not have a joint probability distribution,
we cannot compute the range of values for such moment based on the
expert's beliefs. But the question still remains as to what would
be our best bet given what we know, i.e., what is our best guess for
$E\left(\mathbf{XYZ}\right)$. The quantum mechanical approach does
not address this question, as it is not clear how to get it from the
formalism given that any superposition of the states preferred by
Alice, Bob, and Carlos are acceptable (i.e., we can choose any values
of $c_{A}$, $c_{B}$, and $c_{C}$).

\section{Bayesian Approach\label{sec:Bayesian-Approach}}

Here we focus again on the unknown triple moment. As we mentioned
before, there are many different ways to approach this problem, such
as paraconsistent logics, consensus reaching, or information revision
to restore consistency. Common to all those approaches is the complexity
of how to resolve the inconsistencies, often with the aid of \emph{ad
hoc} assumptions \cite{genest_combining_1986}. Here we show how a
Bayesian approach would deal with the issue \cite{morris_decision_1974,morris_combining_1977}. 

In the Bayesian approach, a decision maker, Deanna ($D$), needs to
access what is the joint probability distribution from a set of inconsistent
expectations. To set the notation, let us first look at the case when
there is only one expert. Let $P_{A}(x)=P_{A}(\mathbf{X}=x|\delta_{A})$
be the probability assigned to event $x$ by Alice conditioned on
Alice's knowledge $\delta_{A}$, and let $P_{D}(x)=P_{D}(\mathbf{X}=x|\delta_{D})$
be Deanna's prior distribution, also conditioned on her knowledge
$\delta_{D}$. Furthermore, let $\mathbf{P}_{A}=P_{A}\left(x\right)$
be a continuous random variable, $\mathbf{P}_{A}\in[0,1],$ such that
its outcome is $P_{A}\left(x\right)$. The idea behind $\mathbf{P}_{A}$
is that consulting an expert is similar to conducting an experiment
where we sample the experts opinion by observing a distribution function,
and therefore we can talk about the probability that an expert will
give an answer for a specific sample point. Then, for this case, Bayes's
theorem can be written as 
\[
P_{D}'\left(x|\mathbf{P}_{A}=P_{A}\left(x\right)\right)=\frac{P_{D}\left(\mathbf{P}_{A}=P_{A}\left(x\right)\right)P_{D}\left(x\right)}{P_{D}\left(\mathbf{P}_{A}=P_{A}\left(x\right)\right)},
\]
where $P_{D}'\left(x|\mathbf{P}_{A}=P_{A}\left(x\right)\right)$ is
Deanna's posterior distribution revised to take into account the expert's
opinion. As is the case with Bayes's theorem, the difficulty lies
on determining the likelihood function $P_{D}\left(\mathbf{P}_{A}\right)$,
as well as the prior. This likelihood function is, in a certain sense,
Deanna's model of Alice, as it is what Deanna believes are the likelihoods
of each of Alice's beliefs. In other words, she should have a model
of the experts. Such model of experts is akin to giving each expert
a certain measure of credibility, since an expert whose model doesn't
fit Deanna's would be assigned lower probability than an expert whose
model fits. 

The extension for our case of three experts and three random variables
is cumbersome but straightforward. For Alice, Bob, and Carlos, Deanna
needs to have a model for each one of them, based on her prior knowledge
about $\mathbf{X}$, \textbf{$\mathbf{Y}$},\textbf{ }and $\mathbf{Z}$,
as well as Alice, Bob, and Carlos. Following Morris \cite{morris_decision_1974},
we construct a set $E$ consisting of our three experts joint priors:
\[
E=\left\{ P_{A}\left(x,y\right),P_{B}\left(y,z\right),P_{C}\left(x,z\right)\right\} .
\]
Deanna's is now faced with the problem of determining the posterior
$P'_{D}\left(x|E\right),$ using Bayes's theorem, given her new knowledge
of the expert's priors. 

In a Bayesian approach, the decision maker should start with a prior
belief on the stocks of $X$, $Y$, and $Z$, based on her knowledge.
There is no recipe for choosing a prior, but let us start with the
simple case where Deanna's lack of knowledge about $X$, $Y$, and
$Z$ means she starts with the initial belief that all combinations
of values for $\mathbf{X}$, $\mathbf{Y}$, and $\mathbf{Z}$ are
equiprobable. Let us use the following notation for the probabilities
of each atom:\hfill $p_{xyz}=P\left(\mathbf{X}=+1,\mathbf{Y}=+1,\mathbf{Z}=+1\right)$,\hfill
$p_{xy\overline{z}}=P\left(\mathbf{X}=+1,\mathbf{Y}=+1,\mathbf{Z}=-1\right)$,\hfill
$p_{\overline{x}y\overline{z}}=P\left(\mathbf{X}=-1,\mathbf{Y}=+1,\mathbf{Z}=-1\right)$,
and so on. Then Deanna's prior probabilities for the atoms are 
\[
p_{xyz}^{D}=p_{\overline{x}yz}^{D}=\cdots=p_{\overline{xyz}}^{D}=\frac{1}{16},
\]
where the superscript $D$ refers to Deanna. 

When reasoning about the likelihood function, Deanna asks what would
be the probable distribution of responses of Alice if somehow she
(Deanna) could see the future (say, by consulting an Oracle) and find
out that $E\left(XY\right)=-1$. For such case, it would be reasonable
for Alice to think it more probable to have, say, $\overline{x}y$
than $xy$, since she was consulted as an expert. So, in terms of
the correlation $\epsilon_{A}$, Deanna could assign the following
likelihood function:
\begin{equation}
P_{D}\left(\epsilon_{A}|\overline{x}y\right)=P_{D}\left(\epsilon_{A}|x\overline{y}\right)=\frac{1}{4}\left(1-\epsilon_{A}\right)^{2},\label{eq:modelPM}
\end{equation}
\begin{equation}
P_{D}\left(\epsilon_{A}|xy\right)=P_{D}\left(\epsilon_{A}|\overline{x}\overline{y}\right)=1-\frac{1}{4}\left(1-\epsilon_{A}\right)^{2},\label{eq:modelPP}
\end{equation}
where the minus sign represents the negative, i.e. $p_{xy\cdot}^{A}=p_{\overline{xy}\cdot}=\frac{1}{4}\left(1+\epsilon_{A}\right)$
and $p_{\overline{x}y\cdot}=p_{x\overline{y}\cdot}=\frac{1}{4}\left(1-\epsilon_{A}\right)$.
Notice that the choice of likelihood function is arbitrary. 

Deanna's posterior, once she knows that Alice thought the correlation
to be zero (cf. (\ref{eq:xy})), constitutes, as we mentioned above,
an experiment. To illustrate the computation, we find its value below,
from Alice's expectation $E_{A}\left(\mathbf{XY}\right)=-1$. From
Bayes's theorem 
\begin{eqnarray*}
p_{xyz}^{D|A} & = & k\left[1-\frac{1}{4}\left(1-\epsilon_{A}\right)^{2}\right]\frac{1}{8}\\
 & = & \frac{1}{4}\left[1-\frac{1}{4}\left(1-\epsilon_{A}\right)^{2}\right]=\frac{3}{16},
\end{eqnarray*}
where the normalization constant $k$ is given by 
\begin{eqnarray*}
k^{-1} & = & \left[1-\frac{1}{4}\left(1-\epsilon_{A}\right)^{2}\right]\frac{1}{8}+\left[\frac{1}{4}\left(1-\epsilon_{A}\right)^{2}\right]\frac{1}{8}+\left[\frac{1}{4}\left(1-\epsilon_{A}\right)^{2}\right]\frac{1}{8}\\
 &  & +\left[1-\frac{1}{4}\left(1-\epsilon_{A}\right)^{2}\right]\frac{1}{8}+\left[\frac{1}{4}\left(1-\epsilon_{A}\right)^{2}\right]\frac{1}{8}+\left[\frac{1}{4}\left(1-\epsilon_{A}\right)^{2}\right]\frac{1}{8}\\
 &  & +\left[1-\frac{1}{4}\left(1-\epsilon_{A}\right)^{2}\right]\frac{1}{8}+\left[1-\frac{1}{4}\left(1-\epsilon_{A}\right)^{2}\right]\frac{1}{8},
\end{eqnarray*}
and we use the notation $p^{D|A}$ to explicitly indicate that this
is Deanna's posterior probability informed by Alice's expectation.
Similarly, we have 
\[
p_{\overline{x}yz}^{D|A}=p_{\overline{x}y\overline{z}}^{D|A}=p_{x\overline{y}z}^{D|A}=p_{x\overline{y}\overline{z}}^{D|A}=\frac{1}{16},
\]
and 
\[
p_{xyz}^{D|A}=p_{xy\overline{z}}^{D|A}=p_{\overline{x}\overline{y}z}^{D|A}=p_{\overline{x}\overline{y}\overline{z}}^{D|A}=\frac{3}{16}.
\]
If we apply Bayes's theorem twice more, to take into account Bob's
and Carlos's opinions given by correlations (\ref{eq:xz}) and (\ref{eq:yz}),
using likelihood functions similar to the one above, we compute the
following posterior joint probability distribution,
\[
p_{xyz}^{D|ABC}=p_{x\overline{y}z}^{D|ABC}=p_{\overline{x}y\overline{z}}^{D|ABC}=p_{\overline{xyz}}^{D|ABC}=0,
\]
\[
p_{\overline{x}yz}^{D|ABC}=p_{x\overline{yz}}^{D|ABC}=\frac{7}{68},
\]
and 
\[
p_{xy\overline{z}}^{D|ABC}=p_{\overline{xy}z}^{D|ABC}=\frac{27}{68}.
\]
Finally, from the joint, we can compute all the moments, including
the triple moment, and obtain $E\left(\mathbf{XYZ}\right)=0$. 

It is interesting to notice that the triple moment from the posterior
is the same as the one from the prior. This is no coincidence. Because
the revisions from Bayes's theorem only modify the values of the correlations,
nothing is changed with respect to the triple moment. In fact, if
we compute Deanna's posterior distribution for any values of the correlations
$\epsilon_{A}$, $\epsilon_{B}$, and $\epsilon_{C}$, we obtain the
same triple moment, as it comes solely from Deanna's prior distribution.
Thus, the Bayesian approach, though providing a proper distribution
for the atoms, does not in any way provide further insights on the
triple moment.

\section{Negative Probabilities\label{sec:Negative-Probabilities}}

We now want to see how we can use negative probabilities to approach
the inconsistencies from Alice, Bob, and Carlos. The first person
to seriously consider using negative probabilities was Dirac in his
Bakerian Lectures on the physical interpretation of relativistic quantum
mechanics \cite{dirac_bakerian_1942}. Ever since, many physicists,
most notably Feynman \cite{feynman_negative_1987}, tried to use them,
with limited success, to describe physical processes (see \cite{muckenheim_review_1986}
or \cite{khrennikov_interpretations_2009} and references therein).
The main problem with negative probabilities is its lack of a clear
interpretation, which limits its use as a purely computational tool.
It is the goal of this section to show that, at least in the context
of a simple example, negative probabilities can provide useful normative
information. 

Before we discuss the example, let us introduce negative probabilities
in a more formal way%
\footnote{We limit our discussion to finite spaces.%
}. Let us propose the following modifications to Kolmogorov's axioms. 
\begin{definition}
Let $\Omega$ be a finite set, ${\cal F}$ an algebra over $\Omega$,
$p$ and $p'$ real-valued functions, $p:{\cal F}\rightarrow\mathbb{R}$,
$p':{\cal F}\rightarrow\mathbb{R}$, and $M^{-}=\sum_{\omega_{i}\in\Omega}\left|p\left(\left\{ \omega_{i}\right\} \right)\right|$.
Then $\left(\Omega,{\cal F},p\right)$ is a negative probability space
if and only if:
\begin{eqnarray*}
\mbox{A.} &  & \forall p'\left(M^{-}\leq\sum_{\omega_{i}\in\Omega}\left|p'\left(\left\{ \omega_{i}\right\} \right)\right|\right)\\
\mbox{B.} &  & \sum_{\omega_{i}\in\Omega}p\left(\left\{ \omega_{i}\right\} \right)=1\\
\mbox{C.} &  & p\left(\left\{ \omega_{i},\omega_{j}\right\} \right)=p\left(\left\{ \omega_{i}\right\} \right)+p\left(\left\{ \omega_{j}\right\} \right),~~i\neq j.
\end{eqnarray*}
\end{definition}
\begin{remark}
\label{In-the-case}If it is possible to define a proper joint probability
distribution, then $M^{-}=0$, and A-C are equivalent to Kolmogorov's
axioms. 
\end{remark}
Going back to our example, we have the following equations for the
atoms.
\begin{equation}
p_{xyz}+p_{\overline{x}yz}+p_{x\overline{y}z}+p_{xy\overline{z}}+p_{x\overline{y}\overline{z}}+p_{\overline{x}y\overline{z}}+p_{\overline{x}\overline{y}z}+p_{\overline{x}\overline{y}\overline{z}}=1,\label{eq:total-prob}
\end{equation}
\begin{equation}
p_{xyz}+p_{\overline{x}yz}+p_{x\overline{y}z}+p_{xy\overline{z}}-p_{x\overline{y}\overline{z}}-p_{\overline{x}y\overline{z}}-p_{\overline{x}\overline{y}z}-p_{\overline{x}\overline{y}\overline{z}}=0,\label{eq:expectation-x}
\end{equation}
\begin{equation}
p_{xyz}+p_{\overline{x}yz}-p_{x\overline{y}z}+p_{xy\overline{z}}-p_{x\overline{y}\overline{z}}+p_{\overline{x}y\overline{z}}-p_{\overline{x}\overline{y}z}-p_{\overline{x}\overline{y}\overline{z}}=0,\label{eq:expectation-y}
\end{equation}
\begin{equation}
p_{xyz}+p_{\overline{x}yz}+p_{x\overline{y}z}-p_{xy\overline{z}}-p_{x\overline{y}\overline{z}}-p_{\overline{x}y\overline{z}}+p_{\overline{x}\overline{y}z}-p_{\overline{x}\overline{y}\overline{z}}=0,\label{eq:expectation-z}
\end{equation}
\begin{equation}
p_{xyz}-p_{\overline{x}yz}-p_{x\overline{y}z}+p_{xy\overline{z}}-p_{x\overline{y}\overline{z}}-p_{\overline{x}y\overline{z}}+p_{\overline{x}\overline{y}z}+p_{\overline{x}\overline{y}\overline{z}}=0,\label{eq:correlation-xy}
\end{equation}
\begin{equation}
p_{xyz}-p_{\overline{x}yz}+p_{x\overline{y}z}-p_{xy\overline{z}}-p_{x\overline{y}\overline{z}}+p_{\overline{x}y\overline{z}}-p_{\overline{x}\overline{y}z}+p_{\overline{x}\overline{y}\overline{z}}=-\frac{1}{2},\label{eq:correlation-yz}
\end{equation}
\begin{equation}
p_{xyz}+p_{\overline{x}yz}-p_{x\overline{y}z}-p_{xy\overline{z}}+p_{x\overline{y}\overline{z}}-p_{\overline{x}y\overline{z}}-p_{\overline{x}\overline{y}z}+p_{\overline{x}\overline{y}\overline{z}}=-1,\label{eq:correlation-xz}
\end{equation}
where (\ref{eq:total-prob}) comes from the fact that all probabilities
must sum to one, (\ref{eq:expectation-x})-(\ref{eq:expectation-z})
from the zero expectations for $\mathbf{X}$, $\mathbf{Y}$, and $\mathbf{Z}$,
and (\ref{eq:correlation-xy})-(\ref{eq:correlation-xz}) from the
pairwise correlations. Of course, this problem is underdetermined,
as we have seven equations and eight unknowns (we don't know the unobserved
triple moment). A general solution to (\ref{eq:total-prob})-(\ref{eq:expectation-z})
is
\begin{eqnarray}
p_{xyz} & =-p_{\overline{x}yz} & =-\frac{1}{8}-\delta,\label{eq:first-line-solution}\\
p_{x\overline{y}z} & =p_{\overline{x}y\overline{z}}= & \frac{3}{16},\label{eq:second-line-solution}\\
p_{xy\overline{z}} & =p_{\overline{x}\overline{y}z}= & \frac{5}{16},\label{eq:third-line-solution}\\
p_{x\overline{y}\overline{z}} & =-p_{\overline{x}\overline{y}\overline{z}}= & -\delta,\label{eq:fourth-line-solution}
\end{eqnarray}
where $\delta$ is a real number. From (\ref{eq:first-line-solution})--(\ref{eq:fourth-line-solution})
it follows that, for any $\delta$, some probabilities are negative.
First, we notice that we can use the joint probability distribution
to compute the expectation of the triple moment, which is $E(\mathbf{XYZ})=-\frac{1}{4}-4\delta.$
Since $-1\leq E\left(\mathbf{XYZ}\right)\leq-1$, it follows that
$-1\frac{1}{4}\leq\delta\leq\frac{3}{4}$. Of course, $\delta$ is
not determined by the lower moments, as we should expect, but axiom
A requires $M^{-}$ to be minimized. So, to minimize $M^{-}$, we
focus only on the terms that contribute to it: the negative ones.
To do so, let us split the problem into several different sections.
Let us start with $\delta\geq0$, which gives $M_{\delta\geq0}^{-}=-\frac{1}{8}-2\delta,$
having a minimum of $-\frac{1}{8}$ when $\delta=0$. For $-1/8\leq\delta<0$,
$M_{-\frac{1}{8}\leq\delta<0}^{-}=\delta-\frac{1}{8}+\delta=-\frac{1}{8},$
which is a constant value. Finally, for $\delta<-1/8$, the mass for
the negative terms is given by $M_{\delta<-\frac{1}{8}}^{-}=\frac{1}{8}-2\delta.$
Therefore, negative mass is minimized when $\delta$ is in the following
range 
\[
-\frac{1}{8}\leq\delta\leq0.
\]
Now, going back to the triple correlation, we see that by imposing
a minimization of the negative mass we restrict its values to the
following range:
\[
-\frac{1}{4}\leq E\left(\mathbf{XYZ}\right)\leq\frac{1}{2}.
\]
But equations (\ref{eq:total-prob})-(\ref{eq:correlation-xz}) and
the fact that the random variables are $\pm1$-valued allow any correlation
between $-1$ and $1$, and we see that the minimization of the negative
mass offers further constraints to a decision maker. 

Before we proceed, we need to address the meaning of negative probabilities,
as well as the minimization of $M^{-}$. We saw from Remark \ref{In-the-case}
that when $M^{-}$ is zero we obtain a standard probability measure.
Thus, the value of $M^{-}$ is a measure of how far $p$ is from a
proper joint probability distribution, and minimizing it is equivalent
to asking $p$ to be as close as possible to a proper joint, while
at the same time keeping the marginals. This point in itself should
be sufficient to suggest some normative use to negative probabilities:
a negative probability (with $M^{-}$ minimized) gives us the most
rational bet we can make given inconsistent information. But the question
remains as to the meaning of negative probabilities. 

To give them meaning, let us redefine the probabilities from $p$
to $p^{*}$ such that $p^{*}\left(\left\{ \omega_{i}\right\} \right)=0$
when $p\left(\left\{ \omega_{i}\right\} \right)\leq0$. It follows
from this redefinition that $\sum_{\omega_{i}\in\Omega}p^{*}\left(\left\{ \omega_{i}\right\} \right)\geq1$.
This newly defined probability would not violate Kolmogorov's nonnegativity
axiom, but instead would violate B above. The $p^{*}$'s corresponds
to de Finetti's upper probability measures, and axiom A above guarantees
that such upper is as close to a proper distribution as possible.
Thus, according to a subjective interpretation, the negative probability
atoms correspond to impossible events, and the positive ones to an
upper probability measure consistent with the marginals. Once again,
the triple moment corresponds to our best bet.

\section{Conclusions\label{sec:Conclusions}}

The quantum mechanical formalism has been successful in the social
sciences. However, one of the questions we raised elsewhere was whether
some minimalist versions of the quantum formalism which do not include
a full version of Hilbert spaces and observables could be relevant
\cite{de_barros_joint_2012}. In this paper we adapted the example
modeled with neural oscillators in \cite{de_barros_joint_2012} to
a different case where each random variable could be interpreted as
outcomes of a market, and where the inconsistencies between the correlations
could be interpreted as inconsistencies between experts' beliefs.
Such inconsistencies result in the impossibility to define a standard
probability measure that allows a decision-maker to select an expectation
for the triple moment. The computation of the triple moment from the
inconsistent information was done in this paper using three different
approaches: Bayesian, quantum-like, and negative probabilities. 

With the Bayesian approach, we showed that not only does it rely on
a model of the experts (the likelihood function), but also that no
new information is gained from it, as the triple moment from the prior
is not changed by the application of Bayes's rules. Therefore, the
Bayesian approach had nothing to say about the triple moment. 

Similar to the Bayesian, the quantum approach also had nothing to
say about the triple moment, as the arbitrariness of choices for quantum
superpositions (without any additional constraints) results in all
values of triple moments being possible. In fact, the quantum approach
above could be similarly implemented using a contextual theory. For
instance, Dzhafarov \cite{dzhafarov_selective_2003} proposes the
use of an extended probability space where different random variables
(say, $\mathbf{X}_{z}$ and $\mathbf{X}_{y}$) are used, and where
we then ask how similar they are to each other (for instance, what
is the value of $P\left(\mathbf{X}_{z}\neq\mathbf{X}_{y}\right)$).
However, as with the quantum case, the meaning given to $P\left(\mathbf{X}=1\right)$
in our example does not fit with this model, as it corresponds to
the expectation of an increase in the stock value of company $X$
in the future, and the $X$ that Alice is talking about is exactly
the same one for Bob and Carlos, as it corresponds to the increase
in the objective value (in the future) of a stock in the same company.
Furthermore, as expected due to its similar features, this approach
has the same problem as the quantum one in terms of dealing with the
triple moment, but it has the advantage of making it clearer what
the problem is: the triple moment does not exist because we have nine
random variables instead of three, as we have three different contexts.

The negative probability approach, on the other hand, led to a nontrivial
constraint to the possible values of the triple moment. When used
as a computational tool, a joint probability distribution, and with
it the triple moment, could be obtained. Together with the minimization
of the negative mass $M^{-}$, this joint leads to a nontrivial range
of possible values for the triple moment. Given the interpretation
of negative probabilities with respect to uppers, it follows that
this range is our best guess as to where the values of the triple
moment should lie, given our inconsistent information. Thus, negative
probabilities provide the decision maker with some normative information
that is unavailable in either the Bayesian or the quantum-like approaches.

\paragraph{Acknowledgments. }

Many of the details about negative probabilities were developed in
collaboration with Patrick Suppes, Gary Oas, and Claudio Carvalhaes
on the context of a seminar held at Stanford University in Spring
2011. I am indebted to them as well as the seminar participants for
fruitful discussions. I also like to thank Tania Magdinier, Niklas
Damiris, Newton da Costa, and the anonymous referees for comments
and suggestions. 

\bibliographystyle{splncs}
\bibliography{QuantumBrain}

\end{document}